\documentclass{birkart}
\usepackage{amsmath,amsfonts,amssymb,bbm,url}
 \newtheorem{thm}{Theorem}[section]
 \newtheorem{cor}[thm]{Corollary}
 \newtheorem{lem}[thm]{Lemma}
 \newtheorem{prop}[thm]{Proposition}
 \theoremstyle{definition}
 
 \theoremstyle{remark}

 \numberwithin{equation}{section}

\newcommand{\cA}{\mathcal{A}}
\newcommand{\cB}{\mathcal{B}}

\newcommand{\cH}{\mathcal{H}}

\newcommand{\cM}{\mathcal{M}}
\newcommand{\cU}{\mathcal{U}}

\newcommand{\Cx}{\mathbb{C}}

\newcommand{\idty}{\mathbbm{1}}
\newcommand{\id}{{\rm id}}
\newcommand{\Tr}{{\rm Tr}}
\newcommand{\ran}{{\rm ran}}

\newcommand{\be}{\begin{equation}}
\newcommand{\ee}{\end{equation}}
\newcommand{\bea}{\begin{eqnarray}}
\newcommand{\eea}{\end{eqnarray}}
\newcommand{\beann}{\begin{eqnarray*}}
\newcommand{\eeann}{\end{eqnarray*}}
\newcommand{\eq}[1]{(\ref{#1})}
\usepackage{color}

\def\condex{{\mathbbm E}}
\def\norm#1{\Vert#1\Vert}

\newcommand\Norm[2][]{\ensuremath{\bigl| \bigl|\, #2\,\bigr| \bigr|_{#1}}}

\newcommand{\cN}{\mathcal{N}}
\begin{document}
%
\title[Local approximation of observables]
 {Local approximation of observables \\ and commutator bounds}
\author[B. Nachtergaele]{Bruno Nachtergaele}
\address{%
Department of Mathematics\\
University of California, Davis\\
One Shields Avenue\\
Davis, CA 95616\\
USA}
\email{bxn@math.ucdavis.edu}

\thanks{Based on work supported in part by the National Science Foundation under grant
DMS-1009502 and the European project COQUIT}
\author[V.B. Scholz]{Volkher B. Scholz}
\address{Institut f\"ur Theoretische Physik\\
Leibniz Universit\"at Hannover\\
Appelstra\ss e 2\\
30167 Hannover\\
Germany}
\email{volkher.scholz@itp.uni-hannover.de}

\author[R.F. Werner]{Reinhard F. Werner}
\address{Institut f\"ur Theoretische Physik\\
Leibniz Universit\"at Hannover\\
Appelstra\ss e 2\\
30167 Hannover\\
Germany}
\email{reinhard.werner@itp.uni-hannover.de}

\keywords{property P, quantum conditional expectation, support of observables, small commutator}

\date{March 10, 2011}

\begin{abstract}
We discuss conditional expectations that can be used as generalizations
of the partial trace for quantum systems with an infinite-dimensional Hilbert
space of states.
\end{abstract}

\maketitle

\section{Introduction}\label{sec:intro}

We denote by $\cB(\cH)$ the bounded linear operators on a Hilbert space
$\cH$, equipped with the operator norm, and for $A,B\in\cB(\cH)$, $[A,B]=AB-BA$ is the
commutator of $A$ and $B$. Let $\cH_1$ and $\cH_2$ be two Hilbert spaces and
$\cH=\cH_1\otimes\cH_2$ their tensor product. In this is note we consider the following
situation.  Suppose $A\in\cB(\cH)$ and $\epsilon \geq 0$ are such that
\be\label{commutator_bound}
\Norm{[A, \idty\otimes B]} \leq\epsilon\Vert A\Vert \Vert B\Vert \quad\mbox{ for all } B\in \cB(\cH_2).
\ee
We will prove that there exists $A'\in\cB(\cH_1)$ such that $\Vert A - A'\otimes \idty\Vert\leq \epsilon \Vert A\Vert$.
The case $\epsilon = 0$ is trivial, since in that case we have $A\in (\idty\otimes\cB(\cH_2))'
=\cB(\cH_1)\otimes\idty$, and therefore there exists $A'\in\cB(\cH_1)$ such that $A=A'\otimes
\idty$.  If $\cH_2$ is finite-dimensional, the result is also well-known. In that case one can take
for $A'$ the normalized partial trace of $A$:
$$
A'=\frac{1}{\dim \cH_2}\Tr_{\cH_2} A.
$$
To see that this choice for $A'$ does the job, it suffices to note that
$$
A^\prime\otimes \idty = \int_{\cU(\cH_2)} dU \, (\idty\otimes U^*)A(\idty\otimes U),
$$
where $dU$ is the Haar measure on the unitary group, $\cU(\cH_2)$, of $\cH_2$.
Then, by the assumption \eq{commutator_bound} one has
\be\label{approximation}
\Vert A^\prime\otimes \idty -A\Vert
\leq \int_{\cU(\cH_2)} dU \, \Norm{(\idty\otimes U^*)[A,(\idty\otimes U)]} \leq \epsilon\Vert A\Vert\, .
\ee
Our direct motivation for extending this result to the general case in which $\cH_2$ is allowed
to be infinite-dimensional stems from the recent applications of Lieb-Robinson bounds
\cite{lieb:1972} to obtaining local approximations of time-evolved observables in quantum
mechanics in the works \cite{bravyi:2006,nachtergaele:2006b,bachmann:2011}.

\section{The main lemma}\label{sec:main}

The existence of the the approximation $A'\in\cB(\cH_1)$ satisfying the error bound
\eq{approximation} is shown in the following lemma. The lemma shows that, as in the
finite-dimensional case, one can take $A'$ to given by a completely positive
linear map $\condex\cB(\cH_1\otimes\cH_2)\to\cB(\cH_1)$ which has the defining
properties of a conditional expectation.

\begin{lem}\label{lem:main}
Let $\cH_1$ and $\cH_2$ be Hilbert spaces. Then there is a completely positive linear map $\condex:\cB(\cH_1\otimes\cH_2)\to\cB(\cH_1)$
with the following properties:
\begin{enumerate}
\item For all $A\in\cB(\cH_1)$, $\condex(A\otimes\idty)=A$;
\item Whenever $A\in\cB(\cH_1\otimes\cH_2)\to\cB(\cH_1)$ satisfies the commutator bound
$$
\Norm{[A,\idty\otimes B]} \leq \epsilon \Vert A\Vert\Vert B\Vert\, \mbox{ for all } B\in\cB(\cH_2).
$$
$\condex(A)\in\cB(\cH_1)$ satisfies the estimate
\be
\Vert \condex(A)\otimes \idty -A\Vert\leq \epsilon\Vert A\Vert;
\nonumber
\ee
\item For all $C,D\in \cB(\cH_1)$ and $A\in \cB(\cH_1\otimes\cH_2)$, we have
\be
\condex(CAD)=C\condex(A)D.
\label{condex}\nonumber
\ee
\end{enumerate}
\end{lem}

\begin{proof}
For any finite dimensional projection $P\in\cB(\cH_2)$ denote by $\cU(P)$ the compact group of unitary operators of the form
$U=(\idty-P)+PUP$, and by $\condex_P$ the averaging operator with the normalized Haar measure $dU$ on $\cU(P)$:
\be
  \condex_P(A)=\int_{\cU(P)}\mkern-10mu dU\ (\idty\otimes U^*)A(\idty\otimes U).
\label{EsubP}\ee
By the argument given in the introduction we have $\norm{A-\condex_P(A)}\leq \epsilon\norm A$
and for $C,D\in \cB(\cH_1)$ and $A\in \cB(\cH_1\otimes\cH_2)$, we have
$\condex_P(CAD)=C\condex_P(A)D$.
Moreover, if $P\geq Q$, we have $\cU(P)\supset\cU(Q)$, and hence
\be[(\idty\otimes U,\condex_P(A)]=0\quad
   \mbox{for}\ P\geq Q\ \mbox{and}\ U\in\cU(Q).
\label{EPQcom}\ee
Now let $(P(\alpha))_{\alpha\in I}$ be a universal subnet of the net of finite dimensional projections
over some directed index set $I$. Then since $\norm{\condex_{P(\alpha)}(A)}\leq\norm{A}$,
the universal subnet is bounded and therefore must be weak-*-convergent. We call the limit
$\condex_\infty(A)$. Clearly then, $\condex_\infty$ is linear, completely positive, leaves every
operator $A\otimes\idty$ fixed, and also satisfies the property \eq{condex}.
Moreover, if $A$ satisfies the commutator bound each
$\condex_P(A)$ lies in the compact $(\epsilon\norm A)$-ball around $A$ and so does the limit.

It remains to prove that we can write $\condex_\infty(A)=\condex(A)\otimes\idty$, i.e., that
$[\idty\otimes B,\condex_\infty(A)]=0$ for all $B\in\cB(\cH_2)$. By taking the limit of Eq. \eq{EPQcom} over $P$ along the chosen net, we find that this is true for any $B\in\cU(Q)$ for any finite dimensional $Q$. But these sets generates a weakly dense subalgebra of $\cB(\cH_2)$, which concludes the proof.
\end{proof}

Note that the map $\condex$ of Lemma \ref{lem:main} is completely positive and
unit preserving and therefore bounded (with $\norm{\condex}=1$) and hence
norm-continuous. Norm-continuity is however not always sufficient in applications.
It is sometimes important that the map $A\mapsto A'$ is continuous with respect to a
different, more suitable topology.  In \cite{bachmann:2011}, e.g.,
the local approximations appear in an integral and continuity is relied on to
insure the integrability of the integrand. Since in Lemma \ref{lem:main}
$A'$ is obtained as a weak cluster point, its continuity properties are not
obvious. Therefore, we consider other maps with the properties of a
conditional expecation 1 and 3 of  Lemma \ref{lem:main},
but with a slightly worse approximation property (to be precise, with the
$\epsilon$ in property 2 of  Lemma \ref{lem:main} replaced by $2\epsilon$)
 and which is continuous with respect to the weak (and $\sigma$-weak) operator
 topology.

\begin{prop}\label{prop:continuous}
Let $\cH_1$ and $\cH_2$ be Hilbert spaces and let $\rho$ be a normal state on $\cB(\cH_2)$. Define the
map $\condex_\rho=\id\otimes\rho$ by $\condex_\rho(A\otimes B)=\rho(B)A$ for all $A\in \cB(\cH_1)$ and $B\in \cB(\cH_2)$. Then, $\condex_\rho$ has the properties 1 and 3 of  Lemma \ref{lem:main}
and, whenever $A\in\cB(\cH_1\otimes\cH_2)$ satisfies the commutator bound
$$
\Vert [A,\idty\otimes B]\Vert \leq \epsilon\Vert A\Vert \Vert B\Vert\, \mbox{ for all } B\in\cB(\cH_2).
$$
we have
\be
\label{eq:continuous}
\Vert \condex_\rho(A) -A\Vert\leq 2\epsilon\Vert A\Vert.
\ee
\end{prop}

\begin{proof}
By the lemma we have
\be
\norm{\condex(A)-\condex_\rho(A)} =\norm{\condex_\rho\Bigl(\condex(A)\otimes\idty-A\Bigr)}
  \leq\norm{\condex(A)\otimes\idty-A}
  \leq\epsilon\norm A.
\nonumber
\ee
Therefore, it follows that
\be
\norm{\condex_\rho(A)\otimes\idty-A}\leq\norm{(\condex(A)-\condex_\rho(A))\otimes\idty}+\norm{\condex(A)\otimes\idty-A}
 \leq2\epsilon\norm A.
\nonumber
\ee
\end{proof}

It is unclear whether the factor $2$ in equation \eqref{eq:continuous} is really needed. Numerical evidence suggests that maybe it is even true with the same bound as in the Lemma. By an approximation argument it would suffice to show this in finite dimension. We tried low dimensional ($3\otimes7$) random matrices $A$, choosing for $\rho$ the state farthest removed from the tracial state, namely a pure one. The random matrices where drawn from the unitarily invariant ensemble. Then $\delta=\norm{\condex_\rho(A)-A}$ is readily computed, and in all cases we found unitary operators $U\in\cB(\cH_2)$ such that $\norm{A-(\idty\otimes U^*)A(\idty\otimes U)}\geq\delta$. We are, of course aware, that this is far from conclusive, since by measure concentration random matrices in high dimension might easily avoid the regions of counterexample with high probability. That is, for most of the cases with respect to the unitarily invariant measure the factor $2$ is not needed, but counterexamples nevertheless might exist.

\section{Application to infinite systems}\label{sec:infinite}

So far, we have discussed two-component systems with a Hilbert space of the form $\cH_1\otimes \cH_2$. In applications the decomposition into two components often  corresponds to selecting a
finite subsystem of an infinite system \cite{bachmann:2011}.

Consider a collection of systems labeled by a countable set $\Gamma$ (e.g., $\Gamma$
is often taken to be the $d$-dimensional hypercubic lattice $\mathbb{Z}^d$.) Associated with
each site $x \in \Gamma$,
there is a quantum system with a Hilbert space $\mathcal{H}_x$. For finite $\Lambda \subset
\Gamma$, we define
\begin{equation} \label{eq:hilale}
\mathcal{H}_{\Lambda} = \bigotimes_{x \in \Lambda} \mathcal{H}_x \quad \mbox{and} \quad
\mathcal{A}_{\Lambda} = \bigotimes_{x \in \Lambda} \mathcal{B}( \mathcal{H}_x)
\end{equation}
where $\mathcal{B}( \mathcal{H}_x)$ denotes the bounded linear operators on
$\mathcal{H}_x$. For $\Lambda_0\subset\Lambda\subset\Gamma$, $\mathcal{A}_{\Lambda_0}$
can be identified in the natural way with $\mathcal{A}_{\Lambda_0}\otimes
 \idty_{\Lambda \setminus \Lambda_0}\subset \mathcal{A}_{\Lambda}$.
One then defines
\begin{equation} \label{eq:localg}
\mathcal{A}_{\rm loc} = \bigcup_{\Lambda \subset \Gamma} \mathcal{A}_{\Lambda}
\end{equation}
 as an inductive limit taken over the net of all finite subsets of $\Gamma$. The completion
 of $\mathcal{A}_{\rm loc}$ with respect to the operator norm is a $C^*$-algebra, which we
 will denote by $\cA_\Gamma$.

The strategy of Proposition \ref{prop:continuous} now allows us to define a family of
maps $\condex_\Lambda$, for finite $\Lambda\subset\Gamma$,
such that $\condex_\Lambda :\cA_\Gamma\to \cA_\Lambda$, in a way compatible
with the embeddings $\cA_{\Lambda_0}\subset\cA_{\Lambda}$,
for $\Lambda_0\subset\Lambda$, i.e., such that
\be\label{compatibility}
\condex_{\Lambda_0}=\condex_{\Lambda_0}\circ\condex_\Lambda, \mbox{ if } \Lambda_0\subset\Lambda.
\ee
We will therefore choose a family of normal states on $\cB(\cH_x)$,
or equivalently, a family of density matrices, $\left( \rho_x \right)_{x\in \Gamma}$ and
let $\rho_\Gamma$ be the corresponding a product state on $\cA_\Gamma$.
For each $\Lambda\subset\Gamma$, let $\rho_{\Lambda^c}$ denote the
restriction of $\rho_\Gamma$ to $\cA_{\Gamma\setminus \Lambda}$.
On $\cA_{\rm loc}$, $\condex_\Lambda$ is then defined by setting
\begin{equation}\label{PiLambda}
\condex_\Lambda=\id_{\cA_\Lambda}\otimes \rho_{\Lambda^c}\, .
\end{equation}
and it is straightforward to see that the $\condex_\Lambda$ defined
in this way satisfy the compatibility property \eq{compatibility}.
All these maps are contractions and extend uniquely to $\cA_\Gamma$
by continuous extension, with preservation of the compatibility property.
Clearly, $\condex_\Lambda$ can be considered as a map $\cA_\Gamma\to \cA_\Gamma$
with $\ran \condex_\Lambda=\cA_\Lambda\subset\cA_\Gamma$. Note that the maps
$\condex_\Lambda$ depend on the choice of normal states $\rho_x$. Since the properties
we are interested in here do not explicitly depend on this choice, we supress it
in the notation.

The following property is a direct consequence of the construction of the $\condex_\Lambda$
and Proposition \ref{prop:continuous}.

\begin{cor}
Let $\Lambda\subset\Gamma$be finite.
Suppose $\epsilon \geq 0$ and $A\in\cA_\Lambda$ are such that
$$
\Norm{[A,\idty\otimes B]} \leq \epsilon\Vert A\Vert \Vert B\Vert\, \mbox{ for all } B\in\cA_{\Gamma\setminus\Lambda}.
$$
Then, with $\condex_\Lambda$ the map defined in \eq{PiLambda}, we have
$\condex_\Lambda(A)\in\cA_\Lambda$ and
\be
\Vert \condex_\Lambda(A) -A\Vert\leq 2\epsilon\Vert A\Vert.
\ee
\end{cor}

We remark that if $\dim \cH_x <\infty$, for all $x\in\Gamma$, i.e., when $\cA_\Gamma$ is
a UHF algebra, we can take the normalized partial trace (maximally mixed state) for each of
the $\rho_x$ and replace $2\epsilon$ by $\epsilon$ by the argument given in the
introduction. In either case, it is easy to construct representations of $\cA_\Gamma$
in which the maps $\condex_\Lambda$ are be represented by weakly continuous maps.
Again, it is an interesting question whether the replacement of the `error' $\epsilon$ by 
$2\epsilon$ is really necessary in order to be able to treat the situation with infinite-dimensional
component systems. 

In the next section we discuss the relation of our construction of a conditional
expectations $\condex$, with the property $P$ introduced by Schwartz almost fifty years ago 
\cite{schwartz:1963}. 

\section{Extension to general von Neumann algebras}

The ideas in the main Lemma can be extended to the wider setting of von Neumann algebras, when we replace the algebra $\cB(\cH_2)$ by a general von Neumann algebra $\cM$ on the Hilbert space $\cH$, which replaces $\cH_1\otimes\cH_2$. As usual, $\cM'$ denotes the commutant of $\cM$, i.e., the von Neumann algebra of bounded operators commuting with $\cM$.

Some of the following equivalences are known deep results. Our addition is the last item. Let us mention that while some implications in the following proposition are only valid in the case of $\cH$ being separable, the others do not depend on this assumption. This will be made clear in the proof.

\begin{prop} Let $\cM\subset\cB(\cH)$ be a von Neumann algebra with trivial center. Then the following properties are equivalent:
\begin{enumerate}
\item $\cM$ is hyperfinite, i.e., the weak closure of an increasing family of matrix algebras all sharing the same identity.
\item $\cM$ has property P \cite{sakai:1971}, i.e., for every $X\in\cB(\cH)$ the weak*-closed convex hull of $\{U^*XU\mid U\in\cM\ \mbox{unitary}\}$ contains an element of $\cM'$.
\item $\cM'$ is injective, i.e., there is a linear map $\condex:\cB(\cH)\to\cM'$ such that $\norm{\condex(X)}\leq\norm X$, and $\condex(A)=A$ for $A\in\cM'$.
\item There is a linear map $\condex:\cB(\cH)\to\cM'$ such that for all $X\in\cB(\cH)$
\begin{equation*}\label{coMbound}
    \norm{\condex(X)-X}\leq\sup\Bigl\{\norm{[X,U]}\Bigm\vert U\in\cM\ \mbox{unitary}\Bigr\} .
\end{equation*}
Furthermore we have that $\condex$ is completely positive with norm $1$ and fulfills
\begin{equation*}
	\condex(AXB)=A\condex(X)B \text{  for  } A,B\in\condex(\cN) .
\end{equation*}
\end{enumerate}

\end{prop}

\begin{proof}
	The next three steps also apply to the case where $\cH$ is not separable.
	
	(1) implies (2) follows by an easy application of the fact that $\cM$ is the weak closure of matrix algebras $\cM_\alpha$, and each of these algebras obviously has property P. But then $\cM$ has property P since its commutant is the intersection of the commutants $\cM'_\alpha$, see also \cite{sakai:1971}, Corollary 4.4.17. 
	
	That (1) implies the first equation in (4) is proven along the lines of the proof of lemma \ref{lem:main}, using again the fact that $\cM$ is the weak closure of matrix algebras $\cM_\alpha$, for which the bound is immediate. But if we choose $X \in \cM'$ we find that $\norm{\condex(X)-X}=0$. The second identity then follows from the fact that if $\condex:\cN\to\cN$ is a projection on a von Neumann algebra such that the range $\condex(\cN)$ is a von Neumann subalgebra containing the identity, then  $\condex$ has to be completely positive with norm $1$, and satisfies the identity $\condex(AXB)=A\condex(X)B$ for $A,B\in\condex(\cN)$ \cite{EvansLewis,tomiyama}. The implication (4) to (3) follows immediately from the last argument. 
	
	The last two implications do require a separable Hilbert space $\cH$. 
	
	Then the equivalence of the notions of hyperfiniteness and injectivity is a deep result by Connes \cite{connes} ( see \cite{haagerup} for a simpler proof). It is easily seen that $\cM$ is injective if and only if $\cM'$ is injective, see \cite{takesaki}, Proposition XV.3.2. Hence, (3) also implies (1). 
	
	The missing implication, \emph{i.e.} (2) implies (3), was proven by Schwartz, in the same paper where he also defined  property P \cite{schwartz:1963}.
\end{proof}

In the above situation we have that, for $A\in\cM$ and $B\in\cM'$, we get $B\condex(A)=\condex(B A)=\condex(A B)=\condex(A)B$, i.e., $\condex(A)\in\cM'\cap\cM''=\Cx\idty$. Hence there is a state $\rho$ such that $\condex(A)=\rho(A)$, and thus
\begin{equation}\label{stateEval}
    \condex(AB)=\rho(A)B\quad\mbox{for\ }A\in\cM,\ B\in\cM'.
\end{equation}
Since the linear hull of the set of elements $AB$ is weak*-dense in $\cB(\cH)$ it would seem that via this formula the state $\rho$ determines $\condex$. However, that is deceptive, because $\condex$ need not be normal (i.e., weak*-continuous). Indeed, the {\it only} case in which $\condex$ is normal, is the case  described in Proposition~\ref{prop:continuous}. The state $\rho$ is then obviously also normal. That $\cM'=\condex(\cB(\cH)))$ must be type one follows from a general result of Tomiyama that the von Neumann type (I, II, or III) cannot increase under normal conditional expectations (see also \cite[Example 1.1]{frank} and \cite[Theorem IV.2.2]{rfwDiss}). Note also that by evaluating with a normal state $\sigma$ we can obtain product states $AB\mapsto\rho(A)\sigma(\condex(B))$ between $\cM$ and $\cM'$ when $\condex$ is normal, such product states could also be made normal, which also entails that $\cM$ is type I \cite{buchholz}.

It follows from this discussion that one can, in general, not use (\ref{stateEval}) to define $\condex$ with a normal state $\rho$: except in the type I case the map $\condex$ densely defined by (\ref{stateEval}) cannot have a continuous extension to $\cB(\cH)$.

\subsection*{Acknowledgments}
The authors gratefully acknowledge the warm hospitality and stimulating atmosphere
of Institut Mittag-Leffler, where the results reported in this paper were obtained during
the program on Quantum Information Theory in Fall 2010. B.N. also acknowledges
the hospitality and support through an ESI Senior Research Fellowship from the
Erwin Schr\"odinger International Institute for Mathematical Physics, Vienna,
during the writing this article.

\end{document}